\documentclass{amsart}

\usepackage{epic}

\usepackage{latexsym}
\usepackage{amsmath}
\usepackage{epsfig}
\usepackage{amssymb}
\usepackage{enumerate}

\usepackage[mathscr]{euscript}

\newtheorem{theorem}{Theorem}[section]
\newtheorem{lemma}[theorem]{Lemma} 

\newtheorem{proposition}[theorem]{Proposition}

\newcommand{\PP}{{\mathbb P}}
\newcommand{\EE}{{\mathbb E}}

\newcommand{\lam}{\lambda}

\begin{document}
\title[Majority rule on Yule trees]{Majority rule has transition ratio 4 on Yule trees under a 2-state symmetric model}
\author{Elchanan Mossel and Mike Steel}
\address{EM: Department of Statistics, UC Berkeley, Berkeley CA; MS Biomathematics Research Centre, University of Canterbury, Christchurch, New Zealand}
\thanks{}

\date{\today}

\begin{abstract}
Inferring the ancestral state at the root of a phylogenetic tree from states observed at the leaves is a problem arising in evolutionary biology. 
The simplest technique -- majority rule -- estimates the root state by the most frequently occurring state at the leaves.  Alternative methods -- such as maximum parsimony - explicitly
take the tree structure into account.  Since either method can outperform the other on particular trees, it is useful to consider the accuracy of the methods on trees
generated under some evolutionary null model, such as a Yule pure-birth model. In this short note, we answer a recently posed question concerning the performance of majority rule on Yule trees under a symmetric 2-state Markovian substitution model of character state change. We show that majority rule is accurate precisely when the ratio of the birth (speciation) rate of the Yule process to the substitution rate exceeds the value $4$.  By contrast, maximum parsimony has been shown to be accurate only when this ratio is at least 6. Our proof relies on a second moment calculation, coupling, and a novel application of a reflection principle. 

\end{abstract}

\maketitle
\section{Introduction}

Given a binary tree, $T$, suppose that a state from some set $S$ is 
assigned uniformly at random to  the root of $T$.
This state then evolves down the tree to the tips of the tree according to a continuous-time Markovian process on $S$, acting along the edges of the tree.  Given the states at the tips of a tree, {\em ancestral state reconstruction}  aims to  estimate the state that was present at the root of the tree.  This problem is particularly relevant to certain questions in evolutionary biology \cite{lib, roy}.

The performance of any ancestral state reconstruction  methods depends on the underlying tree (its topology and branch lengths); accordingly, to compare the performance of different ancestral state reconstruction methods, it is helpful to sample trees from some underlying null distribution.  In evolutionary biology, a natural and widely-used null  process is the 
Yule pure-birth model \cite{sta, yul},  starting with a single lineage at time 0 and grown for time $t$ with birth rate $\lambda$, and this is the model we study here. 
Moreover, for the rest of this paper, we will consider the simple continuous-time Markov process, on the state space $S=\{+1, -1\}$ with an instantaneous substitution rate $m$ between the two states.   Notice that there are two random processes at play here -- the generation of the tree and the substitution process that then applies along the edges of this tree.

 A straightforward information-theoretic argument shows  that any method for estimating the root state at a Yule tree cannot achieve an accuracy that is strictly bounded above $1/2$ as $t$ grows, when $\lambda< 4m$, even when the tree and its branch lengths are given \cite{gas2}.   If just the tree topology is known (but not necessarily its branch lengths) then a natural and often used ancestral state reconstruction approach is to assign a root state that  minimises the number of state changes in the tree required to explain the states at the leaves. This method is known as {\em maximum parsimony} and
it was shown in \cite{gas, li} that when $\lambda/m <6$, this method does no better than guessing the root state, as $t \rightarrow \infty$ (for $\lambda/m>6$, the probability of correct reconstruction (as $t \rightarrow \infty$)  is strictly greater than $1/2$ and converges to $1$ as $\lambda/m \rightarrow \infty$).  The difference between these two ancestral state reconstruction methods is illustrated in Fig.~\ref{fig1}.

\begin{figure}[htb]
\centering
\includegraphics[scale=1.0]{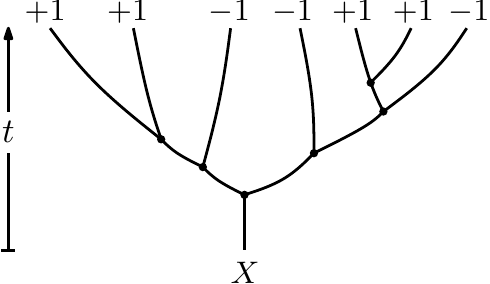}
\caption{A tree generated under a Yule process for time $t$ with seven leaves.  For the states at the leaves shown, majority rule will assign state $+1$ to the unknown root state $X =\pm 1$,
while maximum parsimony will assign state $-1$ to $X$ (since $X=-1$ requires just two state changes in the tree, while $X=+1$ requires  at least three).}
\label{fig1}
\end{figure}  

There is an even simpler way to estimate the root state from the leaf states, which does not even require us to know the tree topology. This is to simply count the number of leaves in each state and use a majority state as the estimate (ties are broken randomly).   For this {\em majority rule} method, the question of determining the ratio of $\lambda/m$ at which
root state estimation retains some accuracy as $t \rightarrow \infty$ was posed in  \cite{gas2}.  In this note, we show that this transition occurs for majority rule at $\lambda/m =4$, which is therefore the smallest possible ratio.  In particular, there is a range ($4< \frac{\lambda}{m}<6$) within which simple majority rule will outperform a recursive method that explicitly uses the tree topology information (maximum parsimony),  despite the fact that for some trees, maximum parsimony can have higher accuracy than majority rule \cite{gas2}.  Our findings are  consistent with simulations that have suggested that majority rule tends to have higher overall accuracy for ancestral state reconstruction on Yule trees than maximum parsimony \cite{gas2}, and complement a  recent study of ancestral state reconstruction on Yule trees for continuous characters evolving under
an Ornstein--Uhlenbeck process \cite{bar}.

It is interesting to compare our results to results on census reconstruction from~\cite{mp}. 
Theorem 1.4 in~\cite{mp} implies that when $\lambda > 4m$, then the reconstruction problem is {\em census solvable}. This means that there is a linear estimator $\sum a_v \sigma_v$ of the root in terms of the leaves $\sigma_v$ which is correlated with the root of the tree. The coefficients of this linear estimator depend on the topology and edge lengths of the tree. In contrast, we are interested in the simpler estimator which is simply given by the majority of the leaf values and show that it results in correlated reconstruction for $\lambda > 4m$. Interestingly, our proof shows that for the Yule tree, the majority reconstruction estimator maximizes the reconstruction probability among all reconstruction methods which are functions of the number of $+1$ and $-1$ leaves only.

\subsection{Preliminaries}

First recall that under the symmetric 2-state process, if the initial state is +1, the state $\sigma_t \in S$ after time $t$ is the random variable  with distribution:
$$\sigma_t= 
\begin{cases}
+1, & \mbox{  with probability } \frac{1}{2}(1+e^{-2mt});\\
-1, &  \mbox{  with probability  }  \frac{1}{2}(1-e^{-2mt}).
\end{cases}
$$
Notice that:
\begin{equation}
\label{ee}
\EE[\sigma_t] = e^{-2mt}.
\end{equation}
Let $L_t$  be the set of leaves at time $t$. It is well known that $N_t:=|L_t|$ has a geometric distribution with parameter $p=e^{-\lambda t}$ (and so $N_te^{-\lambda t}$ converges in distribution to an exponential distribution with mean 1).   In particular, we have:
\begin{equation}
\label{leq}
 \EE[N_t] = e^{\lambda t}.
\end{equation}
Let $$S_t = \sum_{i \in L_t} \sigma_t(i),$$
where $\sigma_t(i)$ is the state at leaf $i$ on the resulting Yule tree, conditional on the root of the tree being in state $+1$.
We first compute the first moment of $S_t$.  Eqn.~(\ref{ee}) gives $\EE[S_t] = \EE[\EE[S_t|N_t]] = \EE[N_t \cdot e^{-2mt}]$
from which Eqn. (\ref{leq}) gives:
\begin{equation}
\label{ee2}
\EE[S_t] = e^{(\lambda -2m)t}.
\end{equation}

\section{Second moment calculation}
Calculating the  second moment of $S_t$ requires a little more care.
First, observe that we may write:  $$S_t^2 = \sum_{i\in L_t} \sigma_t(i)^2 + \sum_{(i,j) \in \tilde{L_t}}\sigma_t(i)\sigma_t(j),$$
where $\tilde{L_t} = \{(i,j) \in L_t \times L_t: i \neq j.\}$
Consequently, since $\sum_{i \in L_t} \sigma_t(i)^2 =N_t$, we have:  
\begin{equation}
\label{sum}
\EE[S_t^2] = e^{\lambda t} + F(t),
\end{equation}
where $F(t) = \EE[\tilde{S_t}]$
for $\tilde{S}_t=\sum_{(i,j) \in \tilde{L_t}}\sigma_t(i)\sigma_t(j).$

Now, suppose that, for the Yule tree grown for time $t$,  two leaves $i$ and $j$ have a most-recent common ancestor at time $t-t'$. Then conditional on this, 
\begin{equation}
\label{helpsout}
\EE[\sigma_t(i)\sigma_t(j)] = e^{-4mt'},
\end{equation}
where expectation is with respect to the substitution process alone.

The function $F(t)$ satisfies $F(0)=0$, and, by the nature of the Yule pure-birth process, and Eqn. (\ref{helpsout}), we have:
\begin{equation}
\label{central}
F(t+\delta) = (1 +2\lambda\delta + O(\delta^2)) \cdot (e^{-4m\delta}F(t)) + (\lambda \delta+ O(\delta^2)) (1-O(\delta)) \EE[N_t]
\end{equation}
Here the first of the two summands considers the extra contribution to $\tilde{S_t}$ between pairs of leaves in the additional $\delta$ period, including the additional contribution when one of the two leaves splits  into two lineages (the $2\lambda \delta$ term -- the probability that both split is $O(\delta^2)$); the second summand deals with the contribution made by the children of any leaf that splits in the $\delta$ period.  

Now, $e^{-4m\delta} = 1-4 m\delta + O(\delta^2)$, so if we apply this, along with  Eqn.(\ref{leq}) in Eqn.~(\ref{central}), and collect together all terms of quadratic or higher order in $\delta$,
we obtain:
$$F(t + \delta) = (1-(4m - 2\lambda)\delta)F(t) + \lambda\delta e^{\lambda t} + O(\delta^2).$$
Rearranging this, and letting  $\delta \rightarrow 0$, we obtain the following linear differential equation for $F(t)$:
\begin{equation}
\label{deq}
\frac{dF}{dt} + 2(2m-\lambda)F = \lambda e^{\lambda t}.
\end{equation}
Solving for $F$ is standard (using the integrating factor $I(s) = e^{(4m-2\lambda) s}$
and the  initial condition $F(0)=0$ gives
$F(t) = e^{(2\lambda - 4m)t} \int_0^t \lambda e^{-( \lambda-4m)s} ds$) and so:
\begin{equation}
\label{deq2}
F(t) =  e^{(2\lambda -4m)t} \times \frac{\lambda}{\lambda - 4m} (1-e^{-(\lambda -4m)t }).
\end{equation}

This, and Eqn. (\ref{ee2}) leads to the following result:

\begin{proposition}
\label{helpsprop}
$\EE[S_t^2] = e^{\lam t} +F(t),$ where $F(t)$ is given by (\ref{deq2}). In particular, when 
$\lambda > 4m$, then for all $t \geq 0$:
$$\frac{\EE[S_t^2]}{\EE[S_t]^2}= e^{-rt}  + \frac{1}{(1-4m/\lambda)}(1 - e^{-rt}),$$
where  $r=\lambda - 4m>0$.

\end{proposition}

We note that exactly the same proof can be applied to $S_{t,\pm}$, which is $S_t$ conditioned on the root being $\pm 1$ 
We will use this to establish our desired result.

\section{A lower bound on the total variation distance of $S_{t,-}$ and $S_{t,+}$}
Out next goal is to show the following.
\begin{lemma} 
\label{lemhelps}
Provided  $\lambda > 4m$,  then for $r=\lambda - 4m>0$ :
\[
d_{TV}(S_{t,+},S_{t,-}) \geq \frac{1-4m/\lambda}{1-4me^{-rt}/\lambda},
\]
and the expression on the right is a monotone decreasing function of $t$, from $1$ (at $t=0$) to $1-\frac{4m}{\lambda}$ (as $t \rightarrow \infty$). 
\end{lemma}

\begin{proof}
We first recall that the total variation distance between any two random variables $X,Y$ on the same probability space $\Omega$ is defined by:
\begin{equation} \label{eq:TVdef1} 
d_{TV}(X,Y) := \frac{1}{2} \sum_{\omega \in \Omega} |\PP[X = \omega]  - \PP[Y = \omega]|.
\end{equation} 
A dual definition, which will be used in the proof of the lemma, is given by: 
\begin{equation} \label{eq:TVdef2} 
d_{TV}(X,Y) := \inf \{ \PP[X' \neq Y'] : X' \sim X, Y' \sim Y \},
\end{equation}
where the infimum is taken over all random variables on $\Omega^2$ with $X'$ having the same distribution as $X$, and with $Y'$ having the same distribution as $Y$ (such $(X',Y')$ is called a coupling 
of $X$ and $Y$).  

Let $(X,Y)$ be a coupling of $S_{t,+}$ and $S_{t,-}$. We will use~(\ref{eq:TVdef2}) to place a lower bound on
the total variation distance by providing a lower bound on $\PP[X \neq Y]$.  
This is a standard application of the second moment method (see e.g.~\cite{lev}, Proposition 7.8).
Using Paley--Zygmund's second moment inequality, one has:
$$\PP[X \neq Y] \geq \frac{(\EE [|X - Y|])^2}{\EE[(X - Y)^2]} .$$
By Jensen's inequality one has
\[
(\EE [|X - Y|])^2 \geq (\EE[X] - \EE[Y])^2 = 4\EE[S_{t,+}]^2.
\]
Moreover,
\[
\EE[(X - Y)^2]  \leq 2 \EE[X^2] + 2 \EE[Y^2] = 4 \times \frac{1}{2} (\EE[S_{t,+}^2] + \EE[S_{t,-}^2]) = 
4 \EE[S_{t,+}^2].
\]
Thus we  have proved that:
\[
\PP(X \neq Y) \geq \frac{\EE[S_{t,+}]^2}{\EE[S_{t,+}^2]},
\]
and Lemma~\ref{lemhelps} now follows from Proposition \ref{helpsprop} (noting that $S_t$ in that proposition is $S_{t,+}$). 
This completes the proof of the lemma.
\end{proof}

\section{Majority reconstruction} 
In order to complete the proof, we will establish the following lemma. 
\begin{lemma} 
\label{helps2}
For all $t \geq 0$, the probability $M$ that majority rule reconstructs the root state correctly  
is given by: 
\[
M = \frac{1}{2} + 
\frac{1}{2} d_{TV}(S_{t,+},S_{t,-}).
\]  
\end{lemma} 

\begin{proof} 
Let $S = S_t$ and let $\sigma$ denote the root value. Then, by rewriting~(\ref{eq:TVdef1}),  we see that:
\begin{equation} \label{eq:D}
D := d_{TV}(S_{t,+},S_{t,-}) = \frac{1}{2} \sum_{s} |\PP[S = s | \sigma = +1]  - \PP[S = s | \sigma = -1]|.
\end{equation}
Moreover, the probability of reconstruction by majority rule is given by: 
\begin{equation} \label{eq:M1}
M = \sum_{s > 0} \PP[S = s] \PP[\sigma = +1 | S = s] +  \sum_{s < 0} \PP[S = s] \PP[\sigma = -1 | S = s] + 
\frac{1}{2} \PP[S = 0].
\end{equation}
Since $\PP[\sigma = +1 | S = s] + \PP[\sigma = -1 | S = s] = 1$, we can rewrite $\PP[\sigma = +1 | S = s]$ 
as $0.5 + 0.5(\PP[\sigma = +1 | S = s]-\PP[\sigma = -1 | S = s])$ and similarly for the other terms. 
We thus obtain the following from~(\ref{eq:M1}):
\begin{align} 
M &= \frac{1}{2} + \frac{1}{2} \sum_{s} \PP[S = s] (\PP[\sigma = +1 | S = s]  - \PP[\sigma = -1 | S = s]){\rm sgn}(s) \\ 
&= \frac{1}{2} + \frac{1}{4} \sum_{s} (\PP[S = s | \sigma = +1]  - \PP[S = s | \sigma = -1]) {\rm sgn}(s).
\label{eq:M2}
\end{align}
Comparing~(\ref{eq:M2}) and~(\ref{eq:D}), we see that in order to prove the lemma, it suffices to show that if $s > 0$ then $\PP[S_t = s | +1] > \PP[S_t = s | -1]$, 
while  if $s < 0$ then $\PP[S_t = s | +1] < \PP[S_t = s | -1]$.

The proof of this follows from the reflection principle. Consider the Markov chain $(N_t,S_t)$ where $N_t$ is the population size. 
Let $T$ be the first stopping time where $S_T = 0$ ($T = \infty$ if it does not happen). Then for $s > 0$, we have (where $\sigma$ is the root state)
\[
\PP[S_t = s | T > t, \sigma = +1] > 0, \quad \PP[S_t = s | T > t, \sigma = -1] = 0, \mbox{ and } 
\]
 \[
 \PP[S_t = s | T \leq t, \sigma = +1] = \PP[S_t = s | T \leq t, \sigma =-1].
\]
From this, it follows that:
\[
\PP[S_t = s | \sigma = +1] > \PP[S_t = s | \sigma = -1],
\]
as needed. The symmetric argument applies when $s < 0$. 

\end{proof}  

Recall that when $\lambda/m <4$ then $\lim_{t \rightarrow \infty} M_t = \frac{1}{2}$ (from \cite{gas2} or \cite{li}). 
We can now state our main result which describes what happens when  $\lambda/m >4$, and whose proof is immediate from Lemmas~\ref{lemhelps} and  \ref{helps2}.

\begin{theorem}
Let  $M_t$ denote the probability that majority rule correctly infers the root state for a Yule tree grown at speciation rate $\lambda$ for time $t$, and
with a character evolved on this tree under a 2-state symmetric process with transition rate $m$, where $\lambda/m>4$. 
Then for all $t\geq 0$: $$M_t  \geq  \frac{1}{2} +  \frac{1}{2}\left(\frac{1-4m/\lambda}{1-4me^{-rt}/\lambda} \right),$$
where the term on the right is monotone decreasing from 1 (at $t=0$) to $1-\frac{4m}{\lambda}$ (as $t \rightarrow \infty$).
In particular, for all finite $t\geq 0$: $$M_t  >  \frac{1}{2} +  \frac{1}{2}\left(1 - \frac{4m}{\lambda}\right).$$
\end{theorem}

\section{Acknowledgements}
The authors thank the Simons Institute at UC Berkeley, where this work was carried out.
E.M would like to acknowledge the support of NSF (grant DMS 1106999  and CCF 1320105) and  ONR (DOD ONR grant N000141110140). 
M.S. would like to thank the NZ Marsden Fund and the Allan Wilson Centre for funding support.

\end{document}